\documentclass[11pt,a4paper]{article}

\usepackage[T1]{fontenc}
\usepackage[english]{babel}

\usepackage{geometry}
\usepackage{graphicx}
\usepackage{amssymb,amsmath,amsthm}
\usepackage{delarray}
\usepackage{enumitem}
\usepackage[x11names]{xcolor}
\usepackage{tikz}
\usepackage{caption,subcaption}
\usepackage{hyperref}
\usepackage{authblk}
\usepackage{stmaryrd}
\usepackage{ifthen}
\usepackage{calc}
\usepackage{float}

\graphicspath{{figures/}}

\theoremstyle{plain}
\newtheorem{theorem}{\textbf{Theorem}}
\newtheorem{lemma}[theorem]{\textbf{Lemma}}

\newtheorem{corollary}[theorem]{\textbf{Corollary}}

\newtheorem{conjecture}[theorem]{\textbf{Conjecture}}
\newtheorem{open}[theorem]{\textbf{Open problem}}
\theoremstyle{definition}
\newtheorem{definition}[theorem]{\textbf{Definition}}

\theoremstyle{remark}

\newtheorem{remark*}{Remark}
\newtheorem{example}[theorem]{\textbf{Example}}

\newcommand{\N}{\mathbb{N}}

\newcommand{\B}{\mathbb{B}}
\newcommand{\1}{\mathtt{1}}
\newcommand{\0}{\mathtt{0}}
\newcommand{\bool}{\{\0,\1\}}
\renewcommand{\O}{\mathcal{O}}
\newcommand{\Poly}{{\mathsf{P}}}
\newcommand{\PPoly}{{\oplus\mathsf{P}}}
\newcommand{\ModkPoly}[1]{{\mathsf{Mod_{#1}P}}}
\newcommand{\ModPoly}{{\mathsf{ModP}}}
\newcommand{\SPoly}{{\#\mathsf{P}}}

\newcommand{\NP}{{\mathsf{NP}}}
\newcommand{\coNP}{{\mathsf{coNP}}}
\newcommand{\DP}{{\mathsf{DP}}}
\newcommand{\PSPACE}{{\mathsf{PSPACE}}}
\newcommand{\coPSPACE}{{\mathsf{coPSPACE}}}

\newcommand{\FPSPACEPoly}{{\mathsf{FPSPACE(poly)}}}
\newcommand{\FPoPSPACE}{{\mathsf{FP}}^{\mathsf{PSPACE}}}
\newcommand{\SPSPACE}{{\mathsf{\#PSPACE}}}

\newcommand{\entiers}[2][2]{\ifthenelse{\equal{#1}{2}}
{\llbracket #2\rrbracket}
{\ifthenelse{\equal{#1}{0}}{
\ifthenelse{\equal{11}{\the\catcode`#2}}
{\{0,\ldots,#2 -1\}}
{\newcounter{finaln}
\setcounter{finaln}{#2 -1} 
\{0,\ldots,\thefinaln \}}}
{\{1,\ldots,#2\}}}}

\renewcommand{\int}[1]{[{#1}]}
\newcommand{\mmjblock}[2]{{#1}_{(#2)}}
\newcommand{\mmjoblock}[2]{{#1}_{\{#2\}}}

\newcommand{\Wl}{(W_i)_{i\in\entiers{\ell}}}
\newcommand{\Sk}{\{S_k\}_{k\in\entiers{s}}}

\newcommand{\BPn}{\mathsf{BP}_n}

\newcommand{\lcm}{\text{lcm}}

\newcommand{\decisionpb}[4]{\fbox{\parbox{.98\textwidth}{
  {#1} ({\bf #2})\\
  {Input:} #3\\
  {Question:} #4
}}}

\title{Complexity of Boolean automata networks\\under block-parallel update 
	modes}
\author[1,2]{K{\'e}vin Perrot}
\author[1,2]{Sylvain Sen{\'e}}
\author[2]{L{\'e}ah Tapin}
\affil[1]{Universit{\'e} publique, Marseille, France}
\affil[2]{Aix-Marseille Univ, CNRS, LIS, Marseille, France}
\date{}

\begin{document}
\renewcommand{\labelitemi}{$\bullet$}
\renewcommand{\labelitemii}{$\bullet$}
\setlist[itemize,enumerate]{nosep}

\maketitle

\begin{abstract}
  Boolean automata networks (\emph{aka} Boolean networks)
  are space-time discrete dynamical systems,
  studied as a model of computation and as a representative model of natural phenomena.
  A collection of simple entities (the automata) update their $0$-$1$ states according to local rules.
  The dynamics of the network is highly sensitive to update modes,
  i.e., to the schedule according to which the automata apply their local rule.
  A new family of update modes appeared recently,
  called block-parallel,
  which is dual to the well studied block-sequential.
  Although basic, it embeds the rich feature of update repetitions among
  a temporal updating period, allowing for atypical asymptotic behaviors.
  In this paper, we prove that it is able to breed complex computations,
  squashing almost all decision problems on the dynamics
  to the traditionally highest (for reachability questions) class $\PSPACE$.
  Despite obtaining these complexity bounds for a broad set of local and global properties,
  we also highlight a surprising gap: bijectivity is still $\coNP$.
\end{abstract}


\section{Introduction}

Boolean automata networks, often called Boolean 
networks for the sake of conciseness, belong to the wide family of automata 
networks, which includes all distributed models of computation 
defined locally by means of individual entities (called automata) which interact with each other over 
discrete time, such as cellular automata for instance, performing global 
computations that emerge only from local interactions.\smallskip

The first examples of Boolean automata networks originate in the 1940s in the 
seminal work of McCulloch and Pitts on neural networks~\cite{J-McCulloch1943}, 
where a finite set of neurons are modelled by automata updating their states 
depending on a local threshold Boolean functions.
This model led the computer science community to deal with its formal 
computational properties, allowing Kleene to introduce the concepts of 
finite automata and regular expressions~\cite{RT-Kleene1951,CL-Kleene1956}
and paving the way toward more algebraic, combinatorial and dynamical 
considerations around linear feedback shift 
registers~\cite{J-Kruskal1954,J-Huffman1959,J-Elspas1959,B-Golomb1967}.
It opened numerous further fundamental questions from complexity and 
computability standpoints from the end of the 1970s.
In parallel of these researches anchored on the foundations of computer 
science, basing himself on Hebb's works~\cite{B-Hebb1949}, Rosenblatt 
introduced the perceptron as a simplification of formal neural 
networks~\cite{J-Rosenblatt1958}, opening the field of concrete artificial 
intelligence. 

Whilst such networks were first and foremost viewed as models of computations, 
the junction between the 1960s and the 1970s saw the appearance of Boolean 
automata networks (where local functions can be any Boolean function) in a 
biological modelling framework. 
Indeed, Kauffman and Thomas emphasised the relevance of these 
networks in the context of gene regulation qualitative 
modelling~\cite{J-Kauffman1969,J-Thomas1973}, which has been repeatedly 
confirmed since the 
1990s~\cite{J-Mendoza1998,C-Akutsu1999,J-Giacomantonio2010,J-Wooten2019}.
\smallskip


Our contribution takes place in a fundamental framework at the frontier between 
theoretical computer science, discrete mathematics, and systems biology. 
When working on Boolean automata networks,
it is utmost important to define the way automata update their state over time
(namely the update mode), in order to obtain a discrete dynamical system.
Indeed:
\begin{itemize}
\item Even if the fixed points obtained under the parallel update mode are 
  fixed points obtained under any other update mode~\cite{B-Goles1990},
  specific update modes may generate additional fixed points
  (e.g., block-parallel).
\item The limit cycles which are not fixed points, obtained under a given 
  update mode, are not necessarily conserved under another update 
  mode~\cite{J-Demongeot2008,J-Goles2008,J-Aracena2009,C-Goles2010,J-Aracena2013}.
\end{itemize}
In other words, a Boolean automata network may admit a large number of 
distinct dynamics (depending on the update mode),
which requires a strong attention, in particular when it is 
employed as a phenomenological model in systems biology, namely a model of a 
genetic regulation network for instance. 
In order to achieve a better understanding on genetic regulation, it 
is impossible to decorrelate it to genetic expression. 
Now, the way genes express during a time laps, i.e., the choice of which
genes are transcribed into their association mRNA portion, is a crucial 
problem for which biologists have no clear answer at present. 
Nonetheless, an emerging insight which has become a full-fledged research 
track concerns chromatin dynamics, namely the way the chromatin node 
``geometrically'' evolves over time inside the cell nucleus, letting the 
externally located DNA portions being the only processed by the 
transcription biological 
machinery~\cite{J-Hansen1992,J-Benecke2006,J-Hubner2010,J-Fierz2019}.

From a theoretical standpoint, advances on chromatin dynamics tend to show that 
genetic expression is neither purely asynchronous nor purely synchronous, 
which paves the way to studies on in-between update modes in our framework. 
Anyhow, considering Boolean automata networks as models of computation or as 
models of real biological networks, the question of the impact of update 
modes on their dynamical behaviours is central and needs particular 
attention.

In this line, this paper aims at studying the peculiar role and impact of a 
family of periodic update modes introduced recently, the block-parallel 
update modes, which have been shown to have interesting and relevant features 
from both formal and applied standpoints~\cite{J-Demongeot2020,C-Perrot2024}, 
in the sense that \emph{(i)} they can generate fixed points which are not 
fixed points of the dynamical system obtained when the underlying network 
evolves synchronously, and \emph{(ii)} they can implement specific biological 
timers which are intrinsically governed by phenomena 
exogenous to regulatory control.

These update modes have been addressed formally in~\cite{C-Perrot2024} from a 
standpoint mixing combinatorial and intrinsic simulation approaches.
Here, we increase our understanding by means of complexity theory,
providing ground results on classical decision 
problems related to fixed points and limit cycles, reachability, \emph{etc}.
These new complexity bounds are particularly relevant,
because they highlight that most decision 
problems known to be $\NP$-complete under block-sequential update modes, such 
as the image/preimage problems, the fixed point problem, the limit cycle 
problem~\cite{J-Floreen1989,C-Bridoux2021,HDR-Perrot2022}, are 
$\PSPACE$-complete under block-parallel update modes.
It suggests that the ``expressivity'' of such update modes
comes at a high cost in terms of simulation,
which strengthens the need for structural results.
For instance, usually, computing a single evolution step of a 
dynamical system, i.e., computing the image of a configuration, is
feasible in polynomial time by evaluating some Boolean circuit; under the 
block-parallel update mode family, this problem becomes $\PSPACE$-hard. 
It follows that most problems are squashed to $\PSPACE$,
but there are notable exceptions, related to bijectivity and steadyness.

In Section~\ref{s:def}, we define formally Boolean automata networks under block-parallel update modes,
and present known results.
Section~\ref{s:complexity} exposes our results, starting with an outline of the constructions
which will be used throughout the proofs.
A first set of classical problems on computing images, preimages, fixed points and limit cycles is
characterized: they all jump from $\NP$ (under block-sequential update modes)
to $\PSPACE$ (under block-parallel update modes).
A second set of problems is studied, and a general bound on the recognition of
functional subdynamics is proved.
Regarding global properties, recognizing bijective dynamics remains $\coNP$-complete,
and recognizing constant dynamics becomes $\PSPACE$-complete.
The case of identity recognition is much subtler, and we provide three incomparable bounds:
a trivial $\coNP$-hardness one, a tough $\ModPoly$-hardness, and a $\FPoPSPACE$-completeness
result derived from the recent literature.
In Section~\ref{s:conclusion}, we summarize the results and expose perspectives.

%

\section{Definitions and state of the art}
\label{s:def}

We denote the set of integers by $\entiers{n} = \entiers[0]{n}$,
the Booleans by $\B = \bool$,
the $i$-th component of a vector $x \in \B^n$ by $x_i \in \B$,
and the restriction of $x$ to domain $I\subset\entiers{n}$ by $x_I \in \B^{|I|}$.
For two graphs $G = (V(G),A(G))$ and $H = (V(H),A(H))$,
we denote by $G \sim H$ when they are isomorphic,
i.e.,~when there is a bijection
$\pi : V(G) \to V(H)$ such that $(x,y) \in A(G) \iff (\pi(x), \pi(y)) \in A(H)$.
We denote by $G\sqsubset H$ when $G$ is a subgraph of $H$,
i.e.,~when $G'$ such that $G'\sim G$ can be obtained from $H$ by vertex and arc deletions.


\paragraph{Boolean automata network} 

A \emph{Boolean automata network} (BAN) is a discrete dynamical system on $\B^n$.
A configuration $x\in\B^n$ associates to each of the $n$ automata among $\entiers{n}$ a Boolean state among $\B$.
The individual dynamics of a each automaton $i\in\entiers{n}$
is described by a local function $f_i:\B^n\to\B$ giving its new state according to the current configuration.
To get a dynamics, one needs to settle the order in which the automata update their state
by application of their local function.
That is, an \emph{update schedule} must be given.
The most basic is the parallel update schedule,
where all automata update their state synchronously at each step,
formally as $f:\B^n\to\B^n$ defined by $\forall x\in\B^n:f(x)=(f_0(x),f_1(x),\dots,f_{n-1}(x))$.
In this work, we concentrate on the block-parallel update schedule,
motivated by the biological context of gene regulatory networks,
where each automaton is a gene and the dynamics give clues on cell phenotypes.
Not all automata will be update simultaneously as in the parallel update mode.
They will instead be grouped by subsets.
For simplicity in defining the local functions of a BAN,
we extend the $f_i:\B^n\to\B$ notation to
subsets $I\subseteq\entiers{n}$ as $f_I:\B^n\to\B^{|I|}$.
We also denote $\mmjblock{f}{I}:\B^n\to\B^n$ the update of automata from subset $I$, defined as:
\[
  \forall i \in \entiers{n}:
  \mmjblock{f}{I}(x)_i =
  \begin{cases}
    f_i(x) & \text{if } i \in I\\
    x_i & \text{otherwise.}
  \end{cases}
\]

\paragraph{Block-sequential update schedule} 
A \emph{block-sequential} update schedule is an \emph{ordered partition} of $\entiers{n}$,
given as a sequence of subsets $\Wl$ where $W_i \subseteq \entiers{n}$
is a \emph{block}.
The automata within a block are updated simultaneously,
and the blocks are updated sequentially.
During one iteration (\emph{step}) of the network,
the state of each automaton is updated exactly once.
The update of each block is called a \emph{substep}.
This update mode received great attention on many aspects.
The concept of the \emph{update digraph} is introduced in~\cite{J-Aracena2009}
and characterized in~\cite{J-Aracena2011}
to capture equivalence classes of block-sequential update schedules (leading to the same dynamics).
Conversions between block-sequential and parallel update schedules are investigated
in~\cite{C-Perrotin2023} (how to parallelize a block-sequential update schedule),
\cite{C-Goles2010} (the preservation of cycles throughout the parallelization process),
and~\cite{C-Bridoux2017} (the cost of sequentialization of a parallel update schedule).


\paragraph{Block-parallel update schedule}
A \emph{block-parallel} update schedule is a \emph{partitioned order} of $\entiers{n}$,
given as a set of subsets $\mu=\Sk$ where $S_k = (i^k_0, \dots, i^k_{n_{k}-1})$
is a sequence of $n_k>0$ elements of $\entiers{n}$ for all $k \in \entiers{s}$, called an \emph{o-block}
(shortcut for \emph{ordered-block}).
Each automaton appears in exactly one o-block.
It follows an idea dual to the block-sequential update mode:
the automata within an o-block are updated sequentially,
and the o-blocks are updated simultaneously.
The set of block-parallel update modes of size $n$ is denoted $\BPn$.
Formally, the update of $f$ under $\mu\in\BPn$ is given by
$\mmjoblock{f}{\mu} : \B^n \to \B^n$ defined, with $\ell=\lcm(n_1,\dots,n_s)$, as:
\[
  \mmjoblock{f}{\mu}(x) = 
  \mmjblock{f}{W_{\ell-1}} \circ
  \dots \circ
  \mmjblock{f}{W_1} \circ
  \mmjblock {f}{W_0}(x)
\]
where for all $i \in \entiers{\ell}$ we define 
$W_i = \{i^k_{i \mod n_k} \mid k \in \int{s}\}$.
In order to compute the set of automata updated at each substep,
it is possible to convert a block-parallel update schedule into
a sequence of blocks of length $\ell$
(which is usually not a block-sequential update schedule,
because repetitions of automaton update may appear~\cite{C-Perrot2024}).
We defined this map as $\varphi$:
\[
  \varphi(\Sk) = 
  \Wl \text{ with } W_i = 
  \{i^k_{i \mod n_k} \mid k \in \int{s}\}\text{.}
\]
An example is given on Figure~\ref{fig:example}.
The parallel update schedule corresponds to the
block-parallel update schedule $\mu_\texttt{par}=\{(i) \mid i \in \entiers{n}\} \in \BPn$,
with $\varphi(\mu_\texttt{par})=(\entiers{n}$),
i.e., a single block containing all automata is updated at each step (there is only one substep).

\begin{figure}
  \centering
  \includegraphics{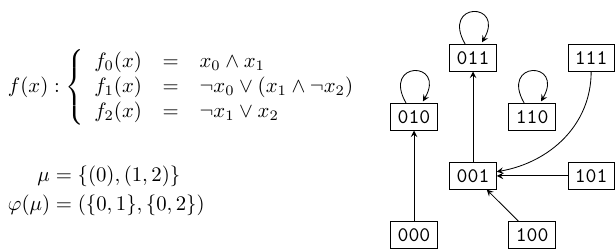}
  \caption{
    Example of an automata network of size $n=3$ with a block-parallel update mode $\mu\in\BPn$.
    Local functions (upper left), conversion of $\mu$ to a sequence of blocks (lower left),
    and dynamics of $\mmjoblock{f}{\mu}$ on configurtion space $\B^3$ (right).
    As an example, in computing the image of configuration $\1\1\1$,
    the first substep (update of automata $0$ and $1$) gives $\1\0\1$,
    and the second substep (update of automata $0$ and $2$) gives $\0\0\1$.
  }
  \label{fig:example}
\end{figure}

Block-parallel update schedules have been introduced in~\cite{J-Demongeot2020},
motivated by applications to gene regulatory networks, and 
their ability to generate new stable configurations
(compared to block-sequential update schedules).
A first theoretical study has been conducted in~\cite{C-Perrot2024},
providing counting formulas and enumeration algorithms,
subject to equivalence relations on the produced dynamics.

\paragraph{Fixed point and limit cycle}
A BAN $f$ of size $n$ under block-parallel update schedule $\mu\in\BPn$
defines a deterministic discrete dynamical system $\mmjoblock{f}{\mu}$ on configuration space $\B^n$.
Since the space is finite, the orbit of any configuration
is ultimately periodic.
For $p\geq 1$, a sequence of configurations $x^0,\dots,x^{p-1}$ is a \emph{limit cycle} of length $p$
when $\forall i\in\entiers{p}:\mmjoblock{f}{\mu}(x^i)=x^{i+1\mod p}$.
For $p=1$ we call $x\in\B^n$ such that $\mmjoblock{f}{\mu}(x)=x$ a \emph{fixed point}.

\paragraph{Complexity}
To be given as input to a decision problem, a BAN is encoded as a tuple of
$n$ Boolean circuits, one for each local function $f_i:\B^N\to\B$ for $i\in\entiers{n}$.
This encoding can be seen as Boolean formulas for each automaton,
and easily implements high-level descriptions with if-then-else statements
(used intensively in our constructions).

The computational complexity of finite discrete dynamical systems has been explored
on the related models of finite cellular automata~\cite{J-Sutner1995}
and reaction networks~\cite{J-Dennunzio2019}.
Regarding automata networks, fixed points received early attention in~\cite{J-Alon1985}
and~\cite{J-Floreen1989}, with existence problems complete for $\NP$.
Because of the fixed point invariance for block-sequential update schedules~\cite{B-Robert1986},
the focus switched to limit cycles~\cite{J-Aracena2013,C-Bridoux2021},
with problems reaching the second level of the polynomial hierarchy.
The interplay of different update schedules has been investigated in~\cite{J-Aracena2013}.
Finaly, let us mention the general complexity lower bounds,
established for any first-order question on the dynamics,
under the parallel update schedule~\cite{C-Gamard2021}.

\section{Computational complexity under block-parallel updates}
\label{s:complexity}

Computational complexity is important to anyone willing to use algorithmic tools
in order to study discrete dynamical systems.
Lower bounds inform on the best worst case time or space one can expect
with an algorithm solving some problem.
The $n$ local functions of a BAN are encoded as Boolean circuits,
which is a convenient formalism corresponding to the high level descriptions one usually employs.
The update mode is given as a list of lists of integers, each of them being encoded either in unary or binary
(this makes no difference, because the encoding of local functions already has a size greater than $n$).

In this section we characterize the computational complexity of typical problems
arising in the framework of automata networks.
We will see that almost all problems reach $\PSPACE$-completeness.
The intuition behind this fact is that the description of a block-parallel update mode
may expend (through $\varphi$) to an exponential number of substeps,
during which a linear bounded Turing machine may be simulated via iterations of a circuit.
We first recall this folklore building block and present a general outline of our constructions
(Subsection~\ref{ss:construction}).
Then we start with results on computing images, preimages, fixed points and limit cycles
(Subsection~\ref{ss:image}),
before studying reachability and global properties of the function $\mmjoblock{f}{\mu}$
computed by an automata network $f$ under block-parallel update schedule $\mu$
(Subsection~\ref{ss:reach}).

\subsection{Outline of the $\PSPACE$-hardness constructions}
\label{ss:construction}

We will design polynomial time many-one reductions from the following
$\PSPACE$-complete decision problem, which appears for example in~\cite{J-Goles2016}.\\[.5em]
\decisionpb{Iterated Circuit Value Problem}{Iter-CVP}
{a Boolean circuit $C:\B^n\to\B^n$, a configuration $x\in\B^n$, and $i\in\entiers{n}$.}
{does $\exists t\in\N:C^t(x)_i=\1$?}

\begin{theorem}[folklore]
  {\normalfont \textbf{Iter-CVP}} is $\PSPACE$-complete.
\end{theorem}

\begin{proof}
  The problem belongs to $\PSPACE$ because the circuit only needs to be iterated at most $2^n$ times
  (tracked with a counter on $n$ bits),
  without keeping track of the history. Indeed, since the space is finite of size $2^n$,
  after that number of iterations the limit behavior has been explored,
  hence either a $\1$ has appeared at position $i$ (positive instance),
  or none will ever appear (negative instance).

  The $\PSPACE$-hardness is obtained by reduction from \textbf{Linear Space Acceptance}~\cite[Section~7.4 page~175]{gareyjohnson},
  which asks whether a given linear bounded Turing machine $M$
  (i.e., a Turing machine using no more working tape cells than its input does on the input tape)
  accepts (or halts on) a given input $w$ of size $\ell$.
  One can construct in polynomial time from $M=(Q,\Gamma,\Sigma,B,q_i,q_f,\delta)$ a circuit with the following inputs and outputs:
  \begin{itemize}
    \item $\ell\,\lceil\log_2(|\Gamma|)\rceil$ bits for the input tape,
    \item $\ell\,\lceil\log_2(|\Gamma|)\rceil$ bits for the working tape,
    \item $\lceil\log_2(|Q|)\rceil$ bits for the internal state,
    \item $\lceil\log_2(\ell)\rceil$ bits for the position of the head,
    \item $1$ bit for signaling the acceptance (or halt),
  \end{itemize}
  such that one iteration of the circuit corresponds to one step of the Turing machine computation.
  The configuration $x$ for \textbf{Iter-CVP} corresponds to the starting configuration of $M$ on $w$,
  and $i$ is the bit signaling the acceptance (or halt).
\end{proof}


Before presenting the general outline of our constructions,
we need a technical lemma related to the generation of primes.

\begin{lemma}\label{lem:primes}
  For all $n \geq 2$, a list of distinct prime integers
  $p_1,p_2,\dots,p_{k_n}$ such that $2\leq p_i < n^2$ and $2^n < \prod_{i=1}^{k_n} p_i < 2^{2n^2}$
  can be computed in time $\O(n^2)$, with $k_n=\lfloor\frac{n^2}{2\ln(n)}\rfloor$.
\end{lemma}

\begin{proof}
  By the prime number theorem, there are approximately $\frac{N}{\ln(N)}$ primes lower than $N$.
  As a consequence, distinct prime integers $p_1,p_2,\dots,p_{k_n}$ with $k_n=\lfloor\frac{n^2}{\ln(n^2)}\rfloor$
  can be computed in time $\O(n^2)$ using Atkin sieve algorithm.
  Since $2 \leq p_i < n^2$, we have $2^{k_n} \leq \prod_{i=1}^{k_n} p_i < n^{2k_n}$.
  It holds that $2^{k_n} = 2^{\lfloor\frac{n^2}{2\ln(n)}\rfloor} > 2^n$, and
  $n^{2k_n} \leq n^{\frac{n^2}{\ln(n)}}$ with
  \[
    \log_2\left(n^\frac{n^2}{\ln(n)}\right) =
    \frac{\frac{n^2}{\ln(n)}}{\log_n(2)} =
    \frac{n^2}{\ln(2)}
  \]
  meaning that $n^{2k_n} \leq 2^{\frac{n^2}{\ln(2)}} < 2^{2n^2}$.
\end{proof}

Our constructions of automata netwoks and block-parallel update schedules
for the computational complexity lower bounds are based on the following.

\begin{definition}\label{def:gn}
  For any $n\geq 2$, let $p_1,p_2,\dots,p_{k_n}$ be the $k_n$ primes given by Lemma~\ref{lem:primes},
  and denote $q_j=\sum_{i=1}^{j}p_i$ their cumulative series for $j$ from $0$ to $k_n$.
  Define the automata network $g_n$ on $q_{k_n}$ automata $\entiers{q_{k_n}}$ with constant $\0$ local functions,
  where the components are grouped in o-blocks of length $p_i$, that is with
  $\mu_n=\bigcup_{i\in\entiers{k_n}}\{(q_i,q_i+1,\dots,q_{i+1}-1)\}$.
\end{definition}

\begin{lemma}\label{lem:gn}
  For any $n\geq 2$, one can compute $g_n$ and $\mu_n$ in time $\O(n^4)$,
  and $|\varphi(\mu_n)|>2^n$.
\end{lemma}

\begin{proof}
  The time bound comes from Lemma~\ref{lem:primes} and the fact that $q_{k_n}$ is in $\O(n^4)$.
  The number of blocks in $\varphi(\mu_n)$ is the least common multiple of its o-block sizes,
  which is the product $\prod_{i=1}^{k_n} p_i$, hence from Lemma~\ref{lem:primes}
  we conclude that it is greater than $2^n$.
\end{proof}

The general idea is now to add some automata to $g_n$
and place them within singletons in $\mu_n$, i.e., each of them in a new o-block of length $1$.
We propose an example implementing a binary counter on $n$ bits.

\begin{example}\label{ex:counter}
  Given $n\geq 2$, consider $g_n$ and $\mu_n$ given by Lemma~\ref{lem:gn}.
  Construct $f$ from $g_n$ by adding $n$ Boolean components $\{q_{k_n},\dots,q_{k_n+n}\}$,
  whose local functions increment a binary counter on those $n$ bits,
  until it freezes to $2^n-1$ (all bits in state $\1$).
  Construct $\mu'$ from $\mu_n$ as $\mu'=\mu_n\cup\bigcup_{i\in\entiers{n}}\{(q_{k_n}+i)\}$,
  so that the counter components are updated at each substep.
  Observe that the pair $f,\mu'$ can be still be computed from $n$ in time $\O(n^4)$.
  Figure~\ref{fig:counter} illustrates an example of orbit for $n=3$,
  and one can notice that $\mmjoblock{f}{\mu'}$ is a constant function sending any $x\in\B^n$ to $\0^{q_{k_n}}\1^n$.
  \begin{figure}
    \centering
    \includegraphics{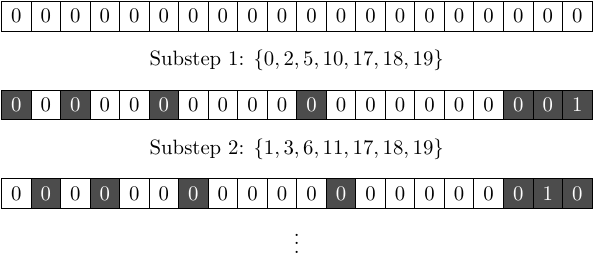}
    \caption{
      Substeps leading to the image of configuration $\0^{q_{k_n}}\0\1\0$ in $\mmjoblock{f}{\mu'}$
      from Example~\ref{ex:counter} for $n=3$ ($k_n=4$ and $q_{k_n}=2+3+5+7=17$).
      The last $3$ bits implement a binary counter, freezing at $7$ ($\1\1\1$).
      Above each substep the block of updated automata is given.
    }
    \label{fig:counter}
  \end{figure}
\end{example}

Remark that we will prove complexity lower bounds by reduction from \textbf{Iter-CVP},
where $n$ will be the number of inputs and outputs of the circuit to be iterated,
hence the integer $n$ itself will be encoded in unary.
As a consequence, the construction of Example~\ref{ex:counter} is computed in polynomial time.

\subsection{Images, preimages, fixed points and limit cycles}
\label{ss:image}

We start the study of the computational complexity of automata networks under block-parallel update schedules
with the most basic problem of computing the image $\mmjoblock{f}{\mu}(x)$
of some configuration $x$ through $\mmjoblock{f}{\mu}$
(i.e., one step of the evolution),
which is already $\PSPACE$-hard.
It is actually hard to compute even a single bit of $\mmjoblock{f}{\mu}(x)$).
We conduct this study as decision problems, where
the fixed point verification is a particular case, still $\PSPACE$-hard.
Recall that the encoding of $\mu$ (with integers in unary or binary) has no decisive influence on the input size,
this latter being characterized by the circuits sizes and in particular their number of inputs,
denoted $n$, which is encoded in unary.
\\[.5em]
\decisionpb{Block-parallel step bit}{BP-Step-Bit}
{$(f_i:\B^n\to\B)_{i\in\entiers{n}}$ as circuits, $\mu\in\BPn$, $x\in\B^n$, $j\in\entiers{n}$.}
{does $\mmjoblock{f}{\mu}(x)_j=1$?}
\\[.5em]
\decisionpb{Block-parallel step}{BP-Step}
{$(f_i:\B^n\to\B)_{i\in\entiers{n}}$ as circuits, $\mu\in\BPn$, $x,y\in\B^n$.}
{does $\mmjoblock{f}{\mu}(x)=y$?}
\\[.5em]
\decisionpb{Block-parallel fixed point verification}{BP-Fixed-Point-Verif}
{$(f_i:\B^n\to\B)_{i\in\entiers{n}}$ as circuits, $\mu\in\BPn$, $x\in\B^n$.}
{does $\mmjoblock{f}{\mu}(x)=x$?}\\

This first set of problems is related to the image of a given configuration $x$,
which allows the reasonings to concentrate on the dynamics of substeps for that single configuration $x$,
regardless of what happens for other configurations.
Note that $n$ will be the size of the \textbf{Iter-CVP} instance,
while the size of the automata network will be $q_{k_n}+\ell'+n+1$.

\begin{theorem}\label{thm:image}
  {\normalfont \textbf{BP-Step-Bit}}, {\normalfont \textbf{BP-Step}} and {\normalfont \textbf{BP-Fixed-Point-Verif}}
  are $\PSPACE$-complete.
\end{theorem}

\begin{proof}
  The problems \textbf{BP-Step-Bit}, \textbf{BP-Step} and \textbf{BP-Fixed-Point-Verif}
  are in $\PSPACE$, with a simple algorithm obtaining $\mmjoblock{f}{\mu}(x)$ by
  computing the least common multiple of o-block sizes and
  then using a pointer for each block throughout
  the computation of that number of substeps
  (each substep evaluates local functions in polynomial time).

  \medskip

  We give a single reduction for the hardness of \textbf{BP-Step-Bit}, \textbf{BP-Step}
  and \textbf{BP-Fixed-Point-Verif},
  where we only need to consider the dynamics of the substeps starting from one configuration $x$.
  Given an instance of \textbf{Iter-CVP} with a circuit $C:\B^n\to\B^n$,
  a configuration $\tilde{x}\in\B^n$ and $i\in\entiers{n}$,
  we apply Lemma~\ref{lem:gn} to construct $g_n,\mu_n$ on automata set $P=\entiers{q_{k_n}}$.
  Automata from $P$ have constant $\0$ local functions,
  and the number of substeps is $\ell=|\varphi(\mu_n)|>2^n$
  thanks to the prime's $\lcm$.
  We define a BAN $f$ by adding:
  \begin{itemize}
    \item $\ell'=\lceil\log_2(\ell)\rceil$ automata numbered $B=\{q_{k_n},\dots,q_{k_n}+\ell'-1\}$,
      implementing a counter that increments modulo $\ell$ at each substep,
      and remains fixed when $x_B$ encodes an integer greater or equal to $\ell$
      (case not considered in this proof);
    \item $n$ automata numbered $D=\{q_{k_n}+\ell',\dots,q_{k_n}+\ell'+n-1\}$, whose local functions
      iterate $C:\B^n\to \B^n$ while the counter is smaller than $\ell-1$,
      and go to state $\tilde{x}$ when the counter reaches $\ell-1$, i.e., with
      \[
        f_D(x)=
        \begin{cases}
          C(x_D)&\text{if }x_B<\ell-1\text{,}\\
          \tilde{x}&\text{otherwise; and}
        \end{cases}
      \]
    \item $1$ automaton numbered $R=\{q_{k_n}+\ell'+n\}$, whose local function
      \[
        f_R(x)=x_R\vee x_{q_{k_n}+\ell'+i}
      \]
      records whether a state $\1$
      appeared at automaton in relative position $i$ within $D$.
  \end{itemize}
  We also add singletons to $\mu_n$ for each of these additional automata, by setting
  \[
    \mu'=\mu_n\cup\bigcup_{j\in B\cup D\cup R}\{(j)\}.
  \]
  Now, consider the dynamics of substeps in computing the image of configuration
  \[
    x=\0^{q_{k_n}}\0^{\ell'}\tilde{x}\0.
  \]
  During the first $\ell-1$ substeps:
  \begin{itemize}
    \item automata $P$ have constant $\0$ local function;
    \item automata $B$ increment a counter from $0$ to $\ell-1$;
    \item automata $D$ iterate circuit $C$ from $\tilde{x}$; and
    \item automaton $R$ records whether the $i$-th bit of $D$ has been in state $\1$ during some iteration.
  \end{itemize}
  During the last substep,
  automata $B$ go back to $\0^n$ because of the modulo, and
  automata $D$ go back to state $\tilde{x}$.
  Since the number of substeps $\ell$ is greater than $2^n$ (Lemma~\ref{lem:gn}),
  the iterations of $C$ search the whole orbit of $\tilde{x}$,
  and at the end of the step automaton $R$ has recorded whether the \textbf{Iter-CVP} instance
  is positive (went to state $\1$) or negative (still in state $\0$).
  The images are respectively
  $y_-=\0^{q_{k_n}}\0^{\ell'}\tilde{x}\0$ or $y_+=\0^{q_{k_n}}\0^{\ell'}\tilde{x}\1$.
  This concludes the reductions, to \textbf{BP-Step-Bit} by asking whether automaton $R$
  (numbered $q_{k_n}+2n$) is in state $\1$,
  to \textbf{BP-Step} by asking whether the image of $x$ is $y_+$,
  and to \textbf{BP-Fixed-Point-Verif} because $y_-=x$ ($\coPSPACE$-hardness).
\end{proof}

As a corollary, the associated functional problem is $\PSPACE$-hard for polynomial time Turing reductions
(not for many-one reductions, as there is no concept of negative instance for total functional problems).

\begin{corollary}\label{cor:image}
  Given an automata network $(f_i:\B^n\to\B)_{i\in\entiers{n}}$ as circuits,
  a block-parallel update schedule $\mu\in\BPn$ and a configuration $x\in\B^n$,
  the problem of computing $\mmjoblock{f}{\mu}(x)$ is in $\FPSPACEPoly$ and $\PSPACE$-hard.
\end{corollary}

Deciding whether a given configuration $y$ has a preimage through $\mmjoblock{f}{\mu}$ is also $\PSPACE$-complete.
The difficulty in this reduction is that we need to take into account the image of every configuration $x$.
We modify the preceding construction by setting automata $D$ to $\tilde{x}$
when the counter $B$ encodes $0$.\\[.5em]
\decisionpb{Block-parallel preimage}{BP-Preimage}
{$(f_i:\B^n\to\B)_{i\in\entiers{n}}$ as circuits, $\mu\in\BPn$, $y\in\B^n$.}
{does $\exists x\in\B^n:\mmjoblock{f}{\mu}(x)=y$?} 

\begin{theorem}\label{thm:preimage}
  {\normalfont \textbf{BP-Preimage}} is $\PSPACE$-complete.
\end{theorem}

\begin{proof}
  The algorithm for \textbf{BP-Preimage} computes the image of each configuration
  (enumerated in polynomial space with a simple counter) using the same procedure as \textbf{BP-Step}
  (Theorem~\ref{thm:image}),
  and decides whether there is some $x$ such that $\mmjoblock{f}{\mu}(x)=y$.

  \medskip

  Given an instance $C:\B^n\to\B^n$, $\tilde{x}\in\B^n$, $i\in\entiers{n}$ of \textbf{Iter-CVP},
  we construct the same block-parallel update schedule $\mu'$ as in the proof of Theorem~\ref{thm:image},
  and modify the local functions of automata $D$ and $R$ as follows:
  \[
    f_D(x)=
    \begin{cases}
      C(\tilde{x})&\text{if }x_B=0\\
      C(x_D)&\text{if }0<x_B<\ell-1\\
      \0^n&\text{otherwise}
    \end{cases}
    \qquad
    f_R(x)=
    \begin{cases}
      \tilde{x}_i&\text{if }x_B=0\\
      x_R\vee x_{q_{k_n}+\ell'+i}&\text{otherwise}
    \end{cases}
  \]
  The purpose is that $D$ iterates the circuit from $\tilde{x}$ when the counter is initialized to $0$,
  and that $R$ records whether the $i$-th bit of $D$ has been in state $\1$
  (including the initial substep).
  We set $y=\0^{q_{k_n}}\0^\ell\0^n\1$.

  If the \textbf{Iter-CVP} instance is positive,
  then we have $\mmjoblock{f}{\mu'}(\0^{q_{k_n}}\0^\ell\0^n\0)=y$
  (automata $B$ go back to $\0^{q_{k_n}}$,
  automata $D$ iterate circuit $C$ from $\tilde{x}$ and end in state $\0^n$,
  and automaton $R$ has recorded that the $i$-th bit of $D$ has been to state $\1$).

  Conversely, if there is a configuration $x$ such that $\mmjoblock{f}{\mu'}(x)=y$,
  then the automata from the counter $B$ must have started in state $x_B=\0^{q_{k_n}}$,
  because of the increment modulo $\ell$ which is the number of substeps.
  We deduce that $D$ iterate circuit $C$ for the whole orbit of $\tilde{x}$ and end in state $\0^n$,
  and that automaton $R$ records the answer to the \textbf{Iter-CVP} instance.
  Since it it ends in state $y_R=\1$ by our assumption that $\mmjoblock{f}{\mu'}(x)=y$,
  we conclude that it is positive.
\end{proof}

Now, we study the computational complexity of problems related to the existence
of fixed points and limit cycles in an automata network under block-parallel update schedule.
Again, we need to consider the image of all configurations,
and have no control on either the start configuration $x$ nor the end configuration $y$
during the dynamics of substeps.
In particular, the counter may be initialized to any value,
and the bit $R$ may already be set to $\1$.
We adapt the previous reductions accordingly.\\[.5em]
\decisionpb{Block-parallel fixed point}{BP-Fixed-Point}
{$(f_i:\B^n\to\B)_{i\in\entiers{n}}$ as circuits, $\mu\in\BPn$.}
{does $\exists x\in\B^n:\mmjoblock{f}{\mu}(x)=x$?}\\[.5em]
\decisionpb{Block-parallel limit cycle of length $k$}{BP-Limit-Cycle-$k$}
{$(f_i:\B^n\to\B)_{i\in\entiers{n}}$ as circuits, $\mu\in\BPn$.}
{does $\exists x\in\B^n:\mmjoblock{f}{\mu}^k(x)=x$?}\\[.5em]
\decisionpb{Block-parallel limit cycle}{BP-Limit-Cycle}
{$(f_i:\B^n\to\B)_{i\in\entiers{n}}$ as circuits, $\mu\in\BPn$, $k\in\N_+$.}
{does $\exists x\in\B^n:\mmjoblock{f}{\mu}^k(x)=x$?}\\[.5em]
On limit cycles we have a family of problems (one for each integer $k$),
and a version where $k$ is part of the input (encoded in binary).
It makes no difference on the complexity.

\begin{theorem}
  \label{thm:fixedpoint}
  {\normalfont \textbf{BP-Fixed-Point}},
  {\normalfont \textbf{BP-Limit-Cycle-$k$}} for any $k\in\N_+$
  and {\normalfont \textbf{BP-Limit-Cycle}} are $\PSPACE$-complete.
\end{theorem}

\begin{proof}
  These problems still belong to $\PSPACE$, because they amount to enumerating configurations
  and computing images by $\mmjoblock{f}{\mu}$,
  which can be done using the procedure for \textbf{BP-Step}
  (Theorem~\ref{thm:image}).

  \medskip

  We start with the hardness proof for the fixed point existence problem,
  and we will then adapt it to limit cycle existence problems.
  Given an instance $C:\B^n\to\B^n$, $\tilde{x}\in\B^n$, $i\in\entiers{n}$ of \textbf{Iter-CVP},
  we construct the same block-parallel update schedule $\mu'$ as in the proof of Theorem~\ref{thm:image},
  and modify the local functions of automata $B$ and $R$ as follows:
  \begin{itemize}
    \item automata $B$ increment a counter modulo $\ell$ at each substep,
      and go to $0$ when the counter is greater than (or equal to) $\ell-1$; and
    \item automaton $R$ records whether a state $\1$ appears at the $i$-th bit of $x_D$,
      and flips when the counter is equal to $\ell-1$, i.e.,
      \[
        f_R(x)=
        \begin{cases}
          x_R\vee x_{q_{k_n}+\ell'+i}&\text{if }x_B<\ell-1\text{,}\\
          \neg x_R&\text{otherwise.}
        \end{cases}
      \]
  \end{itemize}
  Recall that automata $D$ iterate the circuit when $x_B<\ell-1$ and go to $\tilde{x}$ otherwise,
  and that the number $\ell$ of substeps is larger than $2^n$.

  If the \textbf{Iter-CVP} instance is positive, then configuration
  $x=\0^{q_{k_n}}\0^{\ell'}\tilde{x}\0$ is a fixed point of $\mmjoblock{f}{\mu'}$.
  Indeed, during the $\ell$-th and last substep,
  the primes $P$ are still in state $\0^{q_{k_n}}$,
  the counter $B$ goes back to $0$ (state $\0^{\ell'}$),
  the circuit $D$ goes back to $\tilde{x}$,
  and automaton $R$ has recorded the $\1$ which is flipped into state $\0$.

  Conversely, if there is a fixed point configuration $x$,
  then the counter must be at most $\ell-1$ because of the modulo $\ell$ increment.
  Furthermore, automata $D$ will encounter one substep during which it goes to $\tilde{x}$,
  hence the resulting configuration on $D$ will be in the orbit of $\tilde{x}$,
  i.e., $x_D$ is in the orbit of $\tilde{x}$.
  Finally, automaton $R$ will also encounter exactly one substep during which it
  is flipped (when $x_B\geq \ell-1$).
  As a consequence, in order to go back to its initial value $x_R$,
  the state of $R$ must be flipped during another substep,
  which can only happen when it is in state $\0$ and automaton $q_{k_n}+\ell'+i$
  is in state $\1$.
  We conclude that the $i$-th bit of a configuration in the orbit of $\tilde{x}$
  is in state $\1$ during some iteration of the circuit $C$,
  meaning that the \textbf{Iter-CVP} instance is positive.
  Remark that in this case, configuration $\0^{q_{k_n}}\0^{\ell}\tilde{x}\0$ is one of the fixed points.

  \medskip

  For the limit cycle existence problems,
  we modify the construction to let the counter go up to $k\ell-1$.
  Precisely:
  \begin{itemize}
    \item $\ell'=\lceil\log_2(k\ell)\rceil$ automata $B$ implement a binary counter
      which is incremented at each substep, and goes to $0$ when $x_B\geq k\ell-1$;
    \item $n$ automata $D$ iterate the circuit $C$ when $x_B<\ell-1$,
      and go to state $\tilde{x}$ otherwise (no change); and
    \item $1$ automaton $R$ records whether a state $\1$ appears in the $i$-th bit of $x_D$,
      and flips when the counter is equal to $\ell-1$.
  \end{itemize}
  The reasoning is identical to the case $k=1$, except that the counter needs $k$ times $\ell$ substeps,
  i.e., $k$ steps, in order to go back to its initial value.
  As a consequence there is no $x$ and $k'<k$ such that $\mmjoblock{f}{\mu}^{k'}(x)=x$,
  and the dynamics has no limit cycle of length smaller than $k$.
  Remark that when the \textbf{Iter-CVP} instance is positive,
  configurations $(\0^{q_{k_n}}B_i\tilde{x}\0)_{i\in\entiers{k}}$
  with $B_i$ the $\ell'$-bits encoding of $i\ell$
  form one of the limit cycles of length $k$.
  Also remark that the encoding of $k$ in binary within the input has no consequence,
  neither on the $\PSPACE$ algorithm, nor on the polynomial time many-one reduction.
\end{proof}

Remark that our construction also applies to the notion of limit cycle $x^0,\dots,x^{p-1}$
where it is furthermore required that all configurations are different
(this corresponds to having the minimum length $p$): the problem is still $\PSPACE$-complete.

\subsection{Reachability and general complexity bounds}
\label{ss:reach}

In this part, we settle the computational complexity of the classical reachability problem,
which is unsurprisingly still $\PSPACE$-hard by reduction from another model of computation.
In light of what precedes, one may be inclined to think that any problem
related to the dynamics of automata networks under block-parallel update schedules is $\PSPACE$-hard.
We prove that this is partly true with a general complexity bound theorem
on subdynamics existing within $\mmjoblock{f}{\mu}$, based on our previous results on fixed points and limit cycles.
However, we will also prove that a Rice-like complexity lower bound analogous
to the main results of~\cite{C-Gamard2021}, i.e., which would state that
any non-trivial question on the dynamics
(on the functional graph of $\mmjoblock{f}{\mu}$)
expressible in first order logics
is $\PSPACE$-hard, does not hold
(unless a collapse of $\PSPACE$ to the first level of the polynomial hierarchy).
Indeed, we will see that deciding the bijectivity ($\forall x,y\in\B^n:\mmjoblock{f}{\mu}(x)=\mmjoblock{f}{\mu}(y)\implies x=y$) is complete for $\coNP$.
We conclude the section with a discussion on reversible dynamics.\\[.5em]
\decisionpb{Block-parallel reachability}{BP-Reachability}
{$(f_i:\B^n\to\B)_{i\in\entiers{n}}$ as circuits, $\mu\in\BPn$, $x,y\in\B^n$.}
{does $\exists t\in\N:\mmjoblock{f}{\mu}^t(x)=y$?}

\begin{theorem}
  \label{thm:reach}
  {\normalfont \textbf{BP-Reachability}} is $\PSPACE$-complete.
\end{theorem}

\begin{proof}
  The problem belongs to $\PSPACE$, because is can naively be solved by simulating
  the dynamics of $\mmjoblock{f}{\mu}$ starting from configuration $x$, for $2^n$ time steps.

  \medskip

  Reachability problems in cellular automata and related models are known to be $\PSPACE$-complete
  on finite configurations~\cite{J-Sutner1995}.
  We reduce from the reachability problem for reaction systems,
  which can be seen as a particular case of Boolean automata networks,
  and is also known to be $\PSPACE$-complete~\cite{J-Dennunzio2019}.
  Given a reaction system $(S,A)$ where $S$ is a finite set of entities,
  and $A$ is a set of reactions of the form $(R,I,P)$
  where $R$ are the reactants, $I$ the inhibitors and $P$ the products,
  we construct the BAN of size $n=|S|$ with local functions:
  \[
    \forall i\in\entiers{n}: f_i(x)=\bigvee_{\substack{(R,I,P)\in A\\\text{such that }i\in P}}
    \left( \bigwedge_{j\in R}x_j \wedge \bigwedge_{k\in I}\neg x_k \right).
  \]
  A configuration $x\in\B^n$ of the BAN corresponds to a state of the reaction system with
  each automaton indicating the presence or absence of its corresponding entity.
  The parallel evolution of $f$ (under $\mu_\texttt{par}$) is in direct correspondance with the evolution
  of the reaction system.
\end{proof}

From the fixed point and limit cycle theorems in Section~\ref{ss:image},
we now derive that any particular subdynamics is hard to identify within $\mmjoblock{f}{\mu}$
under block-parallel update schedule.
A functional graph is a directed graph of out-degree exactly one,
and we assimilate $\mmjoblock{f}{\mu}$ to its functional graph.
We define a family of problems, one for each functional graph $G$
to find as a subgraph of $\mmjoblock{f}{\mu}$,
and prove that the problem is always $\PSPACE$-hard.
Since $\PSPACE=\coPSPACE$, checking the existence of a subdynamics
is as hard as checking the absence of a subdynamics,
even though the former is a local property whereas the latter is a global property
at the dynamics scale. This is understandable in regard of the fact that
$\PSPACE$ scales everything to the global level
(one can search the whole dynamics in $\PSPACE$),
because verifying that a given set of configurations (a certificate)
gives the subgraph $G$ is difficult (Theorem~\ref{thm:image}).\\[.5em]
\decisionpb{Block-parallel $G$ as subdynamics}{BP-Subdynamics-$G$}
{$(f_i:\B^n\to\B)_{i\in\entiers{n}}$ as circuits, $\mu\in\BPn$.}
{does $G\sqsubset\mmjoblock{f}{\mu}$?}\\[.5em]
Remark that asking whether $G$ appears as a subgraph or as an induced subgraph
makes no difference when $G$ is functional (has out-degree exactly one),
because $\mmjoblock{f}{\mu}$ is also a functional graph:
it is necessarily induced since there is no arc to delete.

\begin{theorem}\label{thm:subdynamics}
  {\normalfont \textbf{BP-Subdynamics-$G$}} is $\PSPACE$-complete for any functional graph $G$.
\end{theorem}

\begin{proof}
  A polynomial space algorithm for \textbf{BP-$G$-Subdynamics} consists in
  enumerating all subsets $S\subseteq\B^n$ of size $|S|=|V(G)|$,
  and test for each whether the restriction of $\mmjoblock{f}{\mu}$ to $S$ is isomorphic to $G$
  (functional graphs are planar hence isomorphism can be decided in logarithmic space~\cite{C-Datta2010}).

  \medskip

  For the $\PSPACE$-hardness, the idea is to choose a fixed point or limit cycle in $G$,
  and make it the decisive element whose existence or not lets $G$ be a subgraph
  of the dynamics or not.
  Since $G$ is a functional graph, it is composed of fixed points and limit cycles,
  with hanging trees rooted into them (the trees are pointing towards their root).
  Let $G(v)$ denote the unique out-neighbor of $v\in V(G)$.

  Let us first assume that $G$ has a limit cycle of length $k\geq 2$,
  or a fixed point with a tree of height greater or equal to $1$ hanging
  (the case where $G$ has only isolated limit cycles is treated thereafter).
  A fixed point is assimilated to a limit cycle of length $k=1$.
  Let $G'$ be the graph $G$ without this limit cycle of size $k$,
  and let $U$ be the vertices of $G'$ without out-neighbor (if $k=1$ then $U\neq\emptyset$).
  We reduce from \textbf{Iter-CVP}, and first compute the $f,\mu$ of size $n$ obtained
  by the reduction from Theorem~\ref{thm:fixedpoint}
  for the problem \textbf{BP-Limit-Cycle-$k$}.
  We have that
  $\mmjoblock{f}{\mu}$ has a limit cycle of length $k$
  on configurations $(\0^{q_{k_n}}B_i\tilde{x}\0)_{i\in\entiers{k}}$
  (or configuration $\0^{q_{k_n}}\0^\ell\tilde{x}\0$ for $k=1$)
  if and only if the \textbf{Iter-CVP} instance is positive.

  We construct $g$ on $n+1$ automata,
  and the update schedule $\mu'$ being the union of $\mu$ with a singleton o-block for the new automaton.
  We assume that $n\geq |V(G)|-k$,
  otherwise we pad $f,\mu$ to that size (with identity local functions for the new automata).
  The idea is that $g$ will consist in a copy of $f$ on the subspace $x_n=\0$,
  and a copy of $G'$ on the subspace $x_n=\1$ where the images of the configurations
  corresponding to the vertices of $U$ will be configurations of the potential limit cycle
  of $\mmjoblock{f}{\mu}$ (in the other subspace $x_n=\0$).
  Other configurations in the subspace $x_n=\1$ will be fixed points.
  Figure~\ref{fig:subdynamics} illustrates the construction.
  Recall that $G$ is fixed,
  and consider a mapping $\alpha:V(G)\to\bool^n$
  such that vertices of the limit cycle of length $k$ are sent
  to the configurations $(\0^{q_{k_n}}B_i\tilde{x}\0)_{i\in\entiers{k}}$ respectively
  (or $\0^{q_{k_n}}\0^\ell\tilde{x}\0$ for $k=1$).
  We define:
  \[
    g(x)=\begin{cases}
      f(x_{\entiers{n}})\0 & \text{if } x_n=0\text{,}\\
      \alpha(G(v))\0 & \text{if } x_n=1 \text{ and } \exists v\in U:\alpha(v)=x_{\entiers{n}}\text{,}\\
      \alpha(G(v))\1 & \text{if } x_n=1 \text{ and } \exists v\in G'\setminus U:\alpha(v)=x_{\entiers{n}}\text{,}\\
      x & \text{otherwise.}
    \end{cases}
  \]
  \begin{figure}
    \centering
    \includegraphics{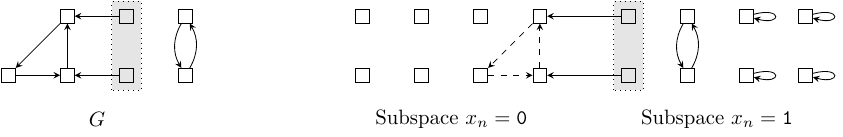}
    \caption{
      Construction of $g$ in the proof of Theorem~\ref{thm:subdynamics}.
      Subspace $x_n=\0$ contains a copy of $f$
      with a potential limit cycle dashed.
      Subspace $x_n=\1$ implements $G'$, and wires configurations of $U$ (grey area)
      to the potential limit cycle in the copy of $f$
      (remaining configurations are fixed points).
    }
    \label{fig:subdynamics}
  \end{figure}
  The obtained dynamics $\mmjoblock{g}{\mu'}$ has one copy of $\mmjoblock{f}{\mu}$
  (on the subspace $x_n=\0$), with a copy of $G'$ (in the subspace $x_n=\1$)
  which becomes a copy of $G$ if configurations
  $(\0^{q_{k_n}}B_i\tilde{x}\0)_{i\in\entiers{k}}$ (or $\0^{q_{k_n}}\0^\ell\tilde{x}\0$ in the case $k=1$)
  form a limit cycle of length $k$.
  Moreover, it becomes a copy of $G$ only if so by our assumption on the limit cycle or fixed point of $G$,
  because the remaining configurations in the subspace $x_n=\1$ are all isolated fixed points.
  This concludes the reduction.

  For the case where $G$ is made of $k$ isolated fixed points,
  we reduce from \textbf{BP-Fixed-Point} and construct an automata network with
  $k$ copies of the dynamics of $f$,
  by adding $\lceil\log_2(k)\rceil$ automata with identity local functions.
\end{proof}

When the property of being a functional graph is dropped,
that is when the out-degree of $G$ is at most one
(otherwise any instance is trivially negative),
problem \textbf{BP-Subdynamics-$G$} is subtler.
Indeed, one can still ask for the existence of fixed points, limit cycles
and any functional subdynamics $\PSPACE$-complete by Theorem~\ref{thm:subdynamics},
but new problems arise, some of which are provably complete only for $\coNP$.
The symmetry of existence versus non existence is broken.
We list some example graphs and the associated properties,
some of which require to forbid two subdynamics.
As mentioned above, non existence of subgraphs leads more naturally
to standard global properties of the dynamics.
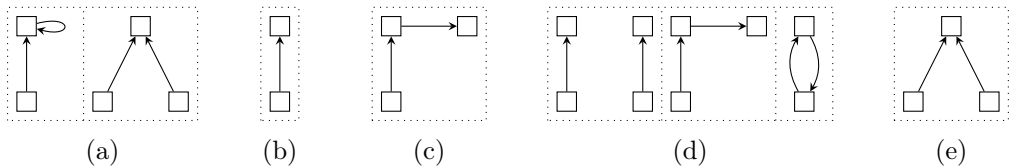
\begin{figure}[h!]
  \centering
  \tikzstyle{node} = [rectangle, draw]
  \tikzstyle{arc} = [-stealth]
  \begin{subfigure}{.2\textwidth}
    \centering
    \begin{tikzpicture}
      \draw[dotted] (-.25,-.25) rectangle (2.25,1.25);
      \draw[dotted] (.75,-.25) -- (.75,1.25);
      \node[node] (n0) at (0,0) {};
      \node[node] (n1) at (0,1) {};
      \draw[arc] (n0) to (n1);
      \draw[arc,loop right] (n1) to (n1);
      \node[node] (nn0) at (1,0) {};
      \node[node] (nn1) at (2,0) {};
      \node[node] (nn2) at (1.5,1) {};
      \draw[arc] (nn0) to (nn2);
      \draw[arc] (nn1) to (nn2);
    \end{tikzpicture}
    \caption{}
    \label{fig:Gsub-bij}
  \end{subfigure}
  \begin{subfigure}{.1\textwidth}
    \centering
    \begin{tikzpicture}
      \draw[dotted] (-.25,-.25) rectangle (.25,1.25);
      \node[node] (n0) at (0,0) {};
      \node[node] (n1) at (0,1) {};
      \draw[arc] (n0) to (n1);
    \end{tikzpicture}
    \caption{}
    \label{fig:Gsub-id}
  \end{subfigure}
  \begin{subfigure}{.15\textwidth}
    \centering
    \begin{tikzpicture}
      \draw[dotted] (-.25,-.25) rectangle (1.25,1.25);
      \node[node] (n0) at (0,0) {};
      \node[node] (n1) at (0,1) {};
      \node[node] (n2) at (1,1) {};
      \draw[arc] (n0) to (n1);
      \draw[arc] (n1) to (n2);
    \end{tikzpicture}
    \caption{}
    \label{fig:Gsub-id2}
  \end{subfigure}
  \begin{subfigure}{.3\textwidth}
    \centering
    \begin{tikzpicture}
      \draw[dotted] (-.25,-.25) rectangle (3.5,1.25);
      \draw[dotted] (1.25,-.25) -- (1.25,1.25);
      \draw[dotted] (2.75,-.25) -- (2.75,1.25);
      \node[node] (n0) at (0,0) {};
      \node[node] (n1) at (0,1) {};
      \node[node] (n2) at (1,0) {};
      \node[node] (n3) at (1,1) {};
      \node[node] (n4) at (1.5,0) {};
      \draw[arc] (n0) to (n1);
      \draw[arc] (n2) to (n3);
      \node[node] (n5) at (1.5,1) {};
      \node[node] (n6) at (2.5,1) {};
      \draw[arc] (n4) to (n5);
      \draw[arc] (n5) to (n6);
      \node[node] (n7) at (3.125,0) {};
      \node[node] (n8) at (3.125,1) {};
      \draw[arc] (n7) to[bend left] (n8);
      \draw[arc] (n8) to[bend left] (n7);
    \end{tikzpicture}
    \caption{}
    \label{fig:Gsub-cst}
  \end{subfigure}
  \begin{subfigure}{.15\textwidth}
    \centering
    \begin{tikzpicture}
      \draw[dotted] (-.25,-.25) rectangle (1.25,1.25);
      \node[node] (n0) at (0,0) {};
      \node[node] (n1) at (1,0) {};
      \node[node] (n2) at (.5,1) {};
      \draw[arc] (n0) to (n2);
      \draw[arc] (n1) to (n2);
    \end{tikzpicture}
    \caption{}
    \label{fig:Gsub-plouf}
  \end{subfigure}
  \caption{Example graphs of out-degree at most one.}
\end{figure}
\begin{enumerate}[label=(\alph*)]
  \item The \emph{non existence} of any of the two graphs from Figure~\ref{fig:Gsub-bij}
    means that the dynamics is bijective (it has only fixed points and limit cycles with no hanging tree),
    expressed in first order as $\forall x,y:\mmjoblock{f}{\mu}(x)=\mmjoblock{f}{\mu}(y)\implies x=y$
    (the space is finite, thus injectivity implies bijectivity).
  \item The \emph{non existence} of the graph from Figure~\ref{fig:Gsub-id}
    means that the dynamics is the identity map on $\B^n$,
    expressed in first order as $\forall x:\mmjoblock{f}{\mu}(x)=x$.
  \item The \emph{non existence} of the graph from Figure~\ref{fig:Gsub-id2}
    means that the dynamics is a union of stars and limit cycles of length two,
    expressed in first order as $\forall x,y,z:\mmjoblock{f}{\mu}(x)\neq y \vee \mmjoblock{f}{\mu}(y)\neq z$.
  \item The \emph{non existence} of any of the three graphs from Figure~\ref{fig:Gsub-cst}
    means that the dynamics is a constant map on $\B^n$,
    expressed in first order as $\forall x,y:\mmjoblock{f}{\mu}(x)=\mmjoblock{f}{\mu}(y)$.
  \item The \emph{non existence} of the graph from Figure~\ref{fig:Gsub-plouf} alone
    is less standard: the dynamics would be a collection of limit cycles,
    and paths leading to fixed points.
\end{enumerate}

\medskip

In what follows, we settle that deciding the bijectivity of $\mmjoblock{f}{\mu}$ is $\coNP$-complete,
and then discuss the complexity of decision problems which are subsets of bijective networks,
such as the problem of deciding whether $\mmjoblock{f}{\mu}$ is the identity.
We conclude the section by proving that it is nevertheless $\PSPACE$-complete to decide whether
$\mmjoblock{f}{\mu}$ is a constant map.
These results hint at the subtleties behind a full characterization of the computational complexity
of \textbf{BP-Subdynamics-$G$} for all graphs of out-degree at most one.\\[.5em]
\decisionpb{Block-parallel bijectivity}{BP-Bijectivity}
{$(f_i:\B^n\to\B)_{i\in\entiers{n}}$ as circuits, $\mu\in\BPn$.}
{is $\mmjoblock{f}{\mu}$ bijective?}\\[.5em]
Remark that, because the space of configurations is finite,
injectivity, surjectivity and bijectivity are equivalent properties of $\mmjoblock{f}{\mu}$.

\begin{lemma}\label{lem:bij}
  Let $f:\B^n\to\B^n$ a BAN and $\mu\in\BPn$ a block-parallel update mode.
  Then $\mmjoblock{f}{\mu}$ is bijective if and only if
  $\mmjblock{f}{W}$ is bijective for every block $W$ of $\varphi(\mu)$.
\end{lemma}

\begin{proof}
  The right to left implication is obvious since $\mmjoblock{f}{\mu}$
  is a composition of bijections $\mmjblock{f}{W}$.
  We prove the contrapositive  of the left to right implication,
  assuming the existence of a block $W$ in $\varphi(\mu)$
  such that $\mmjblock{f}{W}$ is not bijective.
  Let $W_\ell$ be the first such block in the sequence $\varphi(\mu)$,
  so there exist $x,y\in\B^n$ such that $x\neq y$ but $\mmjblock{f}{W_\ell}(x)=\mmjblock{f}{W_\ell}(y)=z$.
  By minimality of $\ell$, the composition $g=\mmjblock{f}{W_{\ell-1}}\circ\dots\circ\mmjblock{f}{W_0}$
  is bijective, hence there also exist $x',y'\in\B^n$ with $x'\neq y'$ such that $g(x')=x$ and $g(y')=y$.
  That is, after the $\ell$-th substep the two configurations $x'$ and $y'$ have the same image $z$,
  and we conclude that $\mmjoblock{f}{\mu}(x')=\mmjoblock{f}{\mu}(y')=
  \mmjblock{f}{W_{p-1}}\circ\dots\circ\mmjblock{f}{W_{\ell+1}}(z)$
  therefore $\mmjoblock{f}{\mu}$ is not bijective.
\end{proof}

Lemma~\ref{lem:bij} shows that bijectivity can be decided at the local level of circuits
(not iterated), which can be checked in $\coNP$ and gives Theorem~\ref{thm:bij}.

\begin{theorem}
  \label{thm:bij}
  {\normalfont \textbf{BP-Bijectivity}} is $\coNP$-complete.
\end{theorem}

\begin{proof}
  A $\coNP$ algorithm can be established from Lemma~\ref{lem:bij},
  because it is equivalent to check the bijectivity at all substeps.
  A non-deterministic algorithm can guess a temporality
  $t\in\entiers{|\varphi(\mu)|}$ (in binary) within the substeps,
  two configurations $x,y$,
  and then check in polynomial time that they certify the non-bijectivity of that substep as follows.
  First, construct $W$ the $t$-th block of $\varphi(\mu)$,
  by computing $t$ modulo each o-block size to get the automata from that o-block.
  Second, check that $\mmjblock{f}{W}(x)=\mmjblock{f}{W}(y)$.

  The $\coNP$-hardness is a direct consequence of that complexity lower bound
  for the particular case of the parallel update schedule~\cite[Theorem~5.17]{HDR-Perrot2022}.
\end{proof}

We now turn our attention to the recognition of identity dynamics.\\[.5em]
\decisionpb{Block-parallel identity}{BP-Identity}
{$(f_i:\B^n\to\B)_{i\in\entiers{n}}$ as circuits, $\mu\in\BPn$.}
{does $\mmjoblock{f}{\mu}(x)=x$ for all $x\in\B^n$?}\\[.5em]
This problem is in $\PSPACE$, and is $\coNP$-hard by reduction
from the same problem in the parallel case~\cite[Theorem~5.18]{HDR-Perrot2022}.
However, it is neither obvious to design a $\coNP$-algorithm to solve it,
nor to prove $\PSPACE$-hardness by reduction from \textbf{Iter-CVP}.

\begin{open}\label{open:id}
  {\normalfont \textbf{BP-Identity}} is $\coNP$-hard and in $\PSPACE$.
  For which complexity class is it complete?
\end{open}

A major obstacle to the design of an algorithm, or of a reduction
from \textbf{Iter-CVP} to \textbf{BP-Identity}, lies in the fact
that, by Theorem~\ref{thm:bij}, ``hard'' instances of the latter are bijective networks
(because the non-bijective instances can be recognized in our immediate lower bound $\coNP$,
and they are all negative instances of \textbf{BP-Identity}).
A reduction would therefore be related to the lengths of cycles in the dynamics of substeps,
and whether they divide the least common multiple of o-block sizes
(for configurations $x\in\B^n$ such that $f(x)=x$) or not ($f(x)\neq x$).

Nonetheless, we are able to prove another lower bound, 
related to the hardness of computing the number of models of a given propositional formula.
The canonical $\ModPoly$-complete problem takes as input a formula $\psi$
and two integers $k$, $i$ encoded in unary, and consists in deciding whether
the number of models of $\psi$ is congruent to $k$ modulo the $i$-th prime number
(which can be computed in polytime).
It generalizes classes $\ModkPoly{k}$ (such as the parity case $\ModkPoly{2}=\PPoly$),
and it is notable that $\SPoly$ polytime truth-table reduces to $\ModPoly$~\cite{J-Kobler1996}.


\begin{theorem}
  {\normalfont \textbf{BP-Identity}} is $\ModPoly$-hard (for polytime many-one reduction).
\end{theorem}

\begin{proof}
  Given a formula $\psi$ on $n$ variables, $m$ and $i$ in unary, 
  we apply Lemma~\ref{lem:gn}
  to construct $g_n,\mu_n$ on automata set $P=\entiers{q_{k_n}}$.
  Automata from $P$ have identity local functions,
  and the number of substeps is $\ell=|\varphi(\mu_n)|>2^n$.
  Let $p_i$ be the $i$-th prime number.
  We add:
  \begin{itemize}
    \item $\ell'=\lceil\log_2(\ell)\rceil$ automata numbered $B=\{q_{k_n},\dots,q_{k_n}+\ell'-1\}$,
      implementing a $\ell'$ bits binary counter that increments
      modulo $\ell$ at each substep, except for configurations with a counter
      greater of equal to $\ell$ which are left unchanged.
    \item $\ell''=\lceil\log_2(p_i)\rceil$ automata numbered $R=\{q_{k_n}+\ell',\dots,q_{k_n}+\ell'+\ell''-1\}$,
      whose local functions are:
      \[
        f_R(x)=\begin{cases}
          x_R - m + 1 \mod p_i &\text{if } x_B=0 \text{ and } x_B \text{ satisfies } \psi\\
          x_R - m \mod p_i &\text{if } x_B=0 \text{ and } x_B \text{ does not satisfy } \psi\\
          x_R + 1 \mod p_i &\text{if } 0<x_B<2^n \text{ and } x_B \text{ satisfies } \psi\\
          x_R &\text{otherwise.}
        \end{cases}
      \]
  \end{itemize}
  We also add singletons to $\mu_n$ for each of these additional automata,
  with $\mu'=\mu_n\cup\bigcup_{j\in B\cup R}\{(j)\}$.
  The resulting dynamics of $\mmjoblock{f}{\mu'}$ proceeds as follows.

  Configurations $x$ such that $x_B\geq\ell$ verify $\mmjoblock{f}{\mu'}(x)=x$,
  because all local functions are identities in this case.
  For configurations $x$ such that $x_B<\ell$,
  during the dynamics of substeps from $x$ to $\mmjoblock{f}{\mu'}(x)$,
  the counter $x_B$ takes exactly once the values from $0$ to $\ell-1$,
  with $\mmjoblock{f}{\mu'}(x)_B=x_B$ (it goes back to its initial value).
  Meanwhile, at each substep with $x_B<2^n$,
  the record of automata $R$ is incremented if and only if $x_B$ satisfies $\psi$,
  with a substraction of $m$  when $x_B=0$.
  Since $\ell>2^n$ each valuation of $\psi$ is checked exactly once,
  and $x_R$ gets added the number of models of $\psi$ minus $m$, modulo $p_i$
  (when $2^n\leq x_B<\ell$ automata $R$ are left unchanged).
  Consequently, we have $\mmjoblock{f}{\mu'}(x)_R=x_R$ if and only if it has been incremented
  $m$ times modulo $p_i$, i.e.,, $f,\mu'$
  is a positive instance of \textbf{BP-Identity} if and only if
  $\psi$, $m$, $i$ is a positive instance of \textbf{Mod-SAT}
  (the number of models of $\psi$ is congruent to $k$ modulo $p_i$).
\end{proof}

Our attemps to prove $\PSPACE$-hardness failed, for the following reasons.
To get bijective circuits one could reduce from reversible Turing machines (RTM) and problem
\textbf{Reversible Linear Space Acceptance}~\cite{J-Lange2000}.
A natural strategy would be to simulate a RTM for an exponential number of subteps,
and then simulate it backwards for that same number of substeps,
while ending in the exact same configuration (identity map) if and only if
the simulation did not halt or was not in the orbit of the given input $w$.
The difficulty with this approach is that the dynamics of substeps must not be the identity map
when a \textit{conjunction} of two temporally separated events happens:
first that the simulation has halted, and second that the starting configuration was $w$.
It therefore requires to remember at least one bit of information,
which is subtle in the reversible setting.
Indeed, the constructions of~\cite{J-Lange2000} and~\cite{B-Morita2017} 
consider only starting configurations of the Turing machine in the initial state and with blank tapes.
However, in the context of Boolean automata networks,
any configuration must be considered (hence any configuration of the simulated Turing machine).

Regarding iterated circuits simulating reversible cellular automata
(for which the whole configuration space is usualy considered),
the literature focuses on
decidability issues~\cite{C-Kari2005,J-Sutner2004},
but a recent contribution fits our setting
and we derive the following.
$\FPoPSPACE$ is the class of functions computable in polynomial time
with an oracle in $\PSPACE$.

\begin{theorem}[{\cite[Theorem~5.7]{A-Eppstein2023}}]\label{thm:rca}
  There is a one-dimensional reversible cellular automaton for which simulating
  any given number of iterations, with periodic boundary conditions, is complete for $\FPoPSPACE$
\end{theorem}

\begin{corollary}
  Given $(f_i:\B^n\to\B)_{i\in\entiers{n}}$ as circuits, $\mu\in\BPn$ such that $\mmjoblock{f}{\mu}$ is bijective,
  $x\in\B^n$ and $t\in\entiers{|\varphi(\mu)|}$ in binary,
  computing the configuration at the $t$-th substep is complete for $\FPoPSPACE$.
\end{corollary}

\begin{proof}
  For a fixed reversible cellular automaton (of any dimension),
  given a configuration of size $n$ and a time $t$,
  one can compute in polynomial time a block-parallel update schedule $\mu$
  and circuits for the local functions of a Boolean automata network
  of large enough size (to encode the CA's state space in binary),
  such that:
  \begin{itemize}
    \item $|\varphi(\mu)|>t$ (by Lemma~\ref{lem:gn};
      these automata are left aside with identity local functions),
    \item one substep of $\mmjoblock{f}{\mu}$ simulates
      one step of the CA; and
    \item $\mmjoblock{f}{\mu}$ is bijective (because the CA is reversible, padding with identity).
  \end{itemize}
  This gives a functional Turing many-one reduction from Theorem~\ref{thm:rca}.
\end{proof}

Intuitively, the dynamics of substeps embeds complexity.
The relationship to the complexity of computing the configuration after the whole step
composed of $|\varphi(\mu)|$ substeps (image through $\mmjoblock{f}{\mu}$),
in order to reach \textbf{BP-Identity}, is not obvious.

\medskip

Being a constant map is another global property of the dynamics,
which turns out to be $\PSPACE$-complete to recognize for BANs
under block-parallel update schedules.\\[.5em]
\decisionpb{Block-parallel constant}{BP-Constant}
{$(f_i:\B^n\to\B)_{i\in\entiers{n}}$ as circuits, $\mu\in\BPn$.}
{does there exist $y\in\B^n$ such that $\mmjoblock{f}{\mu}(x)=y$ for all $x\in\B^n$?}

\begin{theorem}
  \label{thm:cst}
  {\normalfont \textbf{BP-Constant}} is $\PSPACE$-complete.
\end{theorem}

\begin{proof}
  To decide \textbf{BP-Constant}, one can simply enumerate all configurations
  and compute their image (Theorem~\ref{thm:image}) while checking that it always gives the same result.

  \medskip

  For the $\PSPACE$-hardness proof, we reduce from \textbf{Iter-CVP}.
  Given a circuit $C:\B^n\to\B^n$, a configuration $\tilde{x}$ and $i\in\entiers{n}$,
  we apply Lemma~\ref{lem:gn} to construct $g_n,\mu_n$ on automata set $P=\entiers{q_{k_n}}$.
  Automata from $P$ have constant $\0$ local functions,
  and the number of substeps is $\ell=|\varphi(\mu_n)|>2^n$.
  We add (Figure~\ref{fig:cst} illustrates the obtained dynamics):
  \begin{itemize}
    \item $\ell'=\lceil\log_2(\ell)\rceil$ automata numbered $B=\{q_{k_n},\dots,q_{k_n}+\ell'-1\}$,
      implementing a $\ell'$-bits binary counter that increments at each substep,
      and sets all automata from $B$ in state $\1$ when the counter is greater or equal to $\ell-1$;
    \item $n$ automata numbered $D=\{q_{k_n}+\ell',\dots,q_{k_n}+\ell'+n-1\}$, whose local functions
      are given below; and
    \item $1$ automaton numbered $R=\{q_{k_n}+\ell'+n\}$, whose local function is given below.
  \end{itemize}
      \[
        f_D(x)=
        \begin{cases}
          C(\tilde{x})&\text{if }x_B=0\\
          C(x_D)&\text{if }0<x_B<\ell-1\\
          \0^n&\text{otherwise}
        \end{cases}
        \qquad
        f_R(x)=
        \begin{cases}
          \tilde{x}_i&\text{if }x_B=0\\
          x_R\vee x_{q_{k_n}+\ell'+i}&\text{if }0<x_B<\ell\\
          \1&\text{otherwise}
        \end{cases}
      \]
  We also add singletons to $\mu_n$ for each of these additional automata, again by setting
  \[
    \mu'=\mu_n\cup\bigcup_{j\in B\cup D\cup R}\{(j)\}.
  \]

  \begin{figure}
    \centering
    \includegraphics[width=\textwidth]{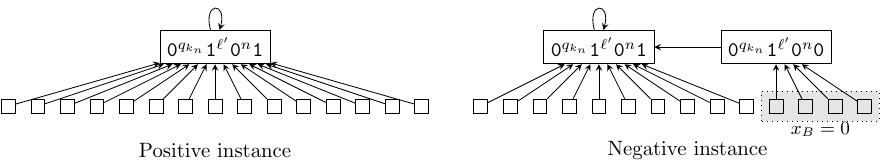}
    \caption{
      Illustration of the dynamics obtained for the reduction
      to \textbf{BP-Constant} in the proof of Theorem~\ref{thm:cst}.
      Configurations $x$ with the counter automata $B$ initialized to $x_B=0$
      either go to $\0^{q_{k_n}}\1^{\ell'}\0^n\1$
      (left, positive instance),
      or to $\0^{q_{k_n}}\1^{\ell'}\0^n\0$ (right, negative instance).
      Only the bit of automata $R$ changes.
    }
    \label{fig:cst}
  \end{figure}

  For any configuration $x$ with a counter not initialized to $0$,
  i.e., with $x_B\neq 0$, the counter will reach and remain in the all $\1$ state
  before the last substep, therefore automata from $D$ will be updated to $\0^n$
  and automaton $R$ will be updated to $\1$.
  We conclude that $\mmjoblock{f}{\mu'}(x)=\0^{q_{k_n}}\1^{\ell'}\0^n\1$.

  For configurations $x$ with $x_B=0$, substeps proceed as follows:
  \begin{itemize}
    \item automata $B$ count until $\ell-1$ at the penultimate substep
      (recall that $\ell=|\varphi(\mu_n)|=|\varphi(\mu'_n)|$),
      which finally brings them all in state $\1$ during the last substep;
    \item automata $D$ iterate the circuit $C$,
      starting from $C(\tilde{x})$ during the first substep; and
    \item automaton $R$ records whether a $\1$ appears or not in the whole orbit of $\tilde{x}$
      (recall that $\ell=|\varphi(\mu'_n)|>2^n$),
      starting from $\tilde{x}$ itself during the first substep
      (even though $x_D\neq\tilde{x}$)
      and without encountering the ``$\1$ otherwise'' case.
  \end{itemize}
  We conclude that the image of $x$ on automata $P$ is $\0^{q_{k_n}}$,
  on $B$ is $\1^{\ell'}$, on $D$ is $\0^n$,
  and on $R$ it depends whether the \textbf{Iter-CVP} instance is positive
  (automaton $R$ in state $\1$) or negative (automaton $R$ in state $\0$).
  
  This completes the reduction: the image is always $\0^{q_{k_n}}\1^{\ell'}\0^n\1$
  if and only if the \textbf{Iter-CVP} instance is positive.
\end{proof}

\section{Conclusion and perspectives}
\label{s:conclusion}

Block-sequential update schedules have a number of substeps limited by
the fact that every automaton is updated only once. Block-parallel update
schedules overcome this restriction, thus significantly raising the
$n$ (number of automata) upper bound for the number of substeps
(Lemma~\ref{lem:gn} gives a backbone construction with more than $2^n$ substeps).
This greatly increases the expressivness of block-parallel dynamics,
and we have demonstrated that this gain in computational power comes along
with higher complexity costs.
A fundamental point is that computing a single transition becomes
$\PSPACE$-hard in this context (Theorem~\ref{thm:image}),
whereas it is feasible in polynomial time for all block-sequential update schedules~\cite{C-Perrotin2023}.

Adapting the constructions in multiple directions,
we were able to prove that almost all classical decision problems
are mashed to $\PSPACE$-completeness
(Subsections~\ref{ss:image} and~\ref{ss:reach}):
verifying that $\mmjoblock{f}{\mu}(x)=y$,
that $\mmjoblock{f}{\mu}(x)=x$, 
deciding the existence of a preimage,
the existence of fixed points and limit cycles,
the existence of a given subdynamics (given as a functionnal digraph),
and deciding whether the dynamics is a constant map.
The problems are in $\Poly$ (single step computations),
$\NP$-complete (existence problems),
or $\coNP$-complete (global dynamical properties)
for block-sequential modes (see~\cite{HDR-Perrot2022}).
In light of these complexity leaps, one might be tempted to extrapolate
to the following conjecture, which turns out to be false (unless a drastic complexity collapse).

\begin{conjecture}[false]
  If a problem is $\NP$-hard or $\coNP$-hard and in $\PSPACE$
  for block-sequential update schedules
  then it is $\PSPACE$-complete for block-parallel update schedules.
\end{conjecture}


However, the bijectivity problem is a special case: as proved by
Lemma~\ref{lem:bij}, a single substep is necessary and sufficient to break 
the bijectivity of the automata network's dynamics.
This prevents the drastic increase in the number of substeps
from affecting the complexity of the bijectivity problem.
Thus, \textbf{BP-Bijectivity} stays in $\coNP$ (Theorem~\ref{thm:bij})
and disproves the above conjecture.
It also disproves the following conjecture (unless a drastic complexity collapse),
because a dynamics is bijective if and only if
it contains neither three distinct configurations
$x,y,z$ such that $\mmjoblock{f}{\mu}(x)=z$ and $\mmjoblock{f}{\mu}(y)=z$,
nor two distinct configurations $x,y$ such that
$\mmjoblock{f}{\mu}(x)=y$ and $\mmjoblock{f}{\mu}(y)=y$
(case $k=2$).

\begin{conjecture}[false]
  For any fixed set of directed graphs $\mathbb{G}=\{G_1,\dots,G_k\}$,
  deciding whether a given BAN $f$ under $\mu\in\BPn$
  contains some $G_i\in\mathbb{G}$ as a subgraph (vertex and arc deletions) is $\PSPACE$-complete.
\end{conjecture}

Again because of bijectivity (does for all configurations $x,y$, we have that
$\mmjoblock{f}{\mu}(x)=\mmjoblock{f}{\mu}(y)$ imples $x=y$?) in $\coNP$,
one cannot transfer the Rice-like complexity bound theorem presented in~\cite{C-Gamard2021},
to the level of $\PSPACE$.
Recognition problems are nonetheless still $\NP$-hard or $\coNP$-hard for non-trivial first order questions,
because parallel is a particular case of block-parallel.

The reachability problem, which is $\PSPACE$-complete for block-sequential modes,
remains $\PSPACE$-complete for block-parallel modes (Theorem~\ref{thm:reach}).
Intuitively, on the one hand the idea of reachability can be embedded in a single transition step
of block-parallel update, because it may have an exponential number of substeps.
On the other hand, the sequence of reachability problems at the level of substeps
combines into a reachability problem at the level of steps which is still in $\PSPACE$.

An important track of research in the domain of automata network consists in
extracting information on the dynamics from the interaction graph,
capturing the architecture of the network:
its vertices are the automata, and its arcs represent effective dependencies among local functions.
In certain cases, solely looking at the interaction graph can even be
sufficient to identify the dynamic. For exemple, with block-sequential update
modes, the only possible interaction graph for the constant dynamics is the one without any
arc, and the only possible interaction graph for the identity dynamics is the one with only
positive loops. This means that recognizing these dynamics is as complex
as determining whether the interaction graph matches the expected pattern
($\coNP$ for constant and $\DP$ for identity).
However, this does not hold as soon as an automaton can be updated more
than once (during a single step of $\mmjoblock{f}{\mu}$),
which is the case for block-parallel update modes.
Indeed, with the exclusive or, it is possible to build an AN of size $5$ that computes the
identity function, yet has non-looping and non-positive arcs in its interaction
graph (see Figure~\ref{fig:xor-id}).
Nevertheless, the exclusive or is linear and thus easy
to predict, but it opens the door towards the possibility of more intricate
examples, to be characterized (Open problem~\ref{open:id}).

\begin{figure}
  \hspace*{-.8cm}
  \includegraphics{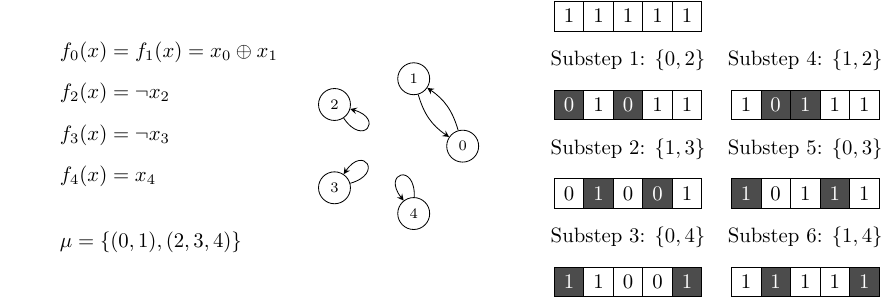}
  \caption{
    A BAN $f$ of size $5$ with non-looping arcs in its interaction graph,
    and a block-parallel update mode $\mu$ such that $\mmjoblock{f}{\mu}$ is the
    identity.
    Left: local functions and update mode.
    Center: interaction graph.
    Right: dynamics of substeps for $\mmjoblock{f}{\mu}(\1\1\1\1\1)=\1\1\1\1\1$.
  }
  \label{fig:xor-id}
\end{figure}

After determining the complexity of recognizing preimages, image points or
fixed points in Subsection~\ref{ss:image}, the next logical step would be the
complexity of counting them. This is not an easy step to make from the
constructions presented in the present work, which are not parcimonious
(for the definition of $\SPSPACE$, see~\cite{J-Ladner1989}).

An important remark for the community is that,
while all the proofs in this paper were written with Boolean automata
networks in mind, the results also hold for non-Boolean automata networks.

Another avenue of research could be questions about the existence of a
block-parallel update schedules verifying a certain property, as
in~\cite{C-Bridoux2021,J-Aracena2013} for block-sequential update schedules.
Given that the fixed point invariance is broken under block-parallel update schedules,
it opens the way for more questions.
The ability to create new fixed points (how and when does it happens?)
is in itself a meaningful track of research.

\section*{Acknowledgment}

This work received support from ANR-18-CE40-0002 FANs,
STIC AmSud CAMA 22-STIC-02 (Campus France MEAE)
and HORIZON-MSCA-2022-SE-01-101131549 ACANCOS projects.

\bibliographystyle{plain}
\bibliography{biblio}

\end{document}